\documentclass[10pt]{amsart}
\usepackage{ObermanLatexPackage}
\usepackage[boxed]{algorithm2e}

\begin{document}
\title[Approximate Convex Hulls]{Approximate Convex Hulls: sketching the convex hull using curvature}
\author{Robert Graham}
\author{Adam~M. Oberman}
\thanks{Department of Mathematics and Statistics, McGill University.  This work was supported by an NSERC Engage grant 2015}
\date{\today}
\begin{abstract}
Convex hulls are fundamental objects in computational geometry.  In moderate dimensions or for large numbers of vertices, computing the convex hull can be  impractical due to the computational complexity of convex hull algorithms.  In this article we approximate the convex hull in using a scalable algorithm which finds high curvature vertices with high probability.   The algorithm is particularly effective for approximating convex hulls which have a relatively small number of extreme points. 
\end{abstract}


\maketitle

\newcommand{\CH}{\text{CH}}
\newcommand{\red}[1]{{\color{red}{#1}}}


%

\section{Introduction}

Computing the convex hull of points is a fundamental problem in the field of computational geometry \cite{de2000computational}.
In moderate dimensions or for large numbers of vertices, computing the exact convex hull can be computationally impractical.  Even the vertex redundancy problem, which computes the extreme points without the full geometric structure of the convex hull is impractical 

Even relatively modern algorithms \cite{clarkson1994more} \cite{matouvsek1992linear} break down or are too expensive for high dimensional problems.  For practical purposes, in moderate dimensions, the computational obstruction is not with the algorithm, but with the convex hull itself.  However, in moderate or high dimensions, the combinatorial structure of high dimensional data sets can be very different from the geometrical intuition obtained from studying low dimensional simplices \cite{barany1989intrinsic} \cite{reitzner2005central}.

Nevertheless, there are many practical problems which are well represented by sampling a large number of high dimensional points, where we expect that either: (i) the number of extreme points is small or, more generally, (ii) the data is contained in a set which is the convex hull of a small number of points.   The second case applies, for example, to data sampled uniformly from a simplex in high dimensions.  This problem occurs hyperspectral data analysis \cite{winter1999n} \cite{winter2000comparison} \cite{boardman1993automating}. 

\begin{figure}[ht]
\includegraphics[scale= .35]{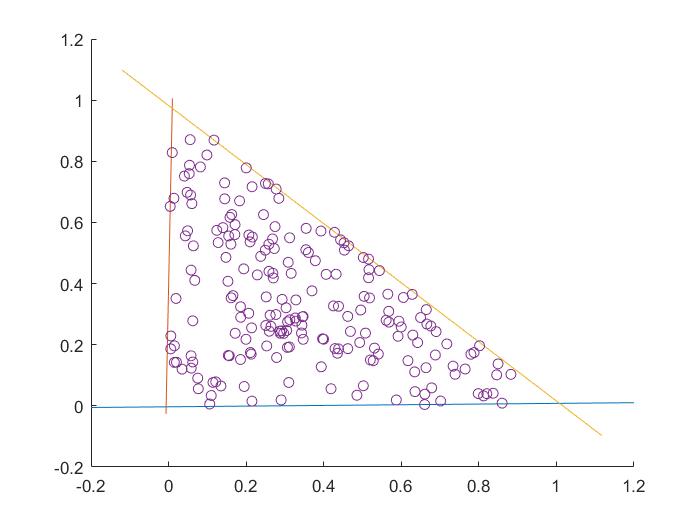}
\caption{Recovering the simplex from sampled points}
\label{fig:hypcomp_2dim}
\end{figure}
In figure \ref{fig:hypcomp_2dim}, we give an illustration of this problem, and of our algorithm.  Here 200 points were sampled uniformly from the 2 dimensional simplex. We used the hyperplane compression algorithm to approximate the original simplex.

Our work is motivated by an application in data reduction.  The algorithm was used to reduce the number of points in a helicopter flight test of Bell Helicopter Model 505 Jet Ranger X~\cite{MaximeCH}.
Approximately twenty million load vectors were recorded during the helicopter certification flight test. The goal
was to extract a small number of extremal loads to be applied on a flight representative fatigue test of the helicopter tailboom assembly. The
exact convex hull contained approximately two thousand points. The algorithm
identified approximately 200 high curvature points. These points were clustered
into a smaller number of points, and finally, using strain response of the tailboom to the load vectors, six load vectors were selected for the fatigue test. This robust data reduction method proved to be effective in identifying extremal loads and was instrumental for the timely certification of the recently released Bell Helicopter 505 Jet Ranger X.

The key idea in our approach is geometrically intuitive. Suppose we have a finite collection of points $X\subset \mathbb{R}^n$ for which we want to approximate the convex hull, denoted $CH(X)$. Given any unit vector $d\in S^{n-1}$, any extreme point
\[
x_d \in  \argmax \{x^\intercal d  \mid  x\in X\}
\]
is a vertex of $CH(X)$, moreover $x^\intercal d \leq x_d^\intercal d$ defines a  supporting hyperplane. Thus if we perform such computations for a large number of unit vectors, we obtain both a collection of vertices whose hull is contained in $CH(X)$ and a collection of hyperplanes whose intersection contains $CH(X)$. These determine inner and outer approximations to $CH(X)$. 

However, if we return to the example of points sampled from a simplex, near a vertex, there can be a high number of extreme points.   Our goal is to reduce the number of extreme points without introducing too large of an error in the convex hull.   The influential Pixel Purity Index (PPI) algorithm~\cite{boardman1993automating} keeps only vertices which are the extremal for more unit vectors, $d$.   In this article, we justify and refine the algorithm using high dimensional curvature concepts. 
The connection with curvature has been hinted at in \cite{theiler2000using}, where it was observed that points which maximize many direction vectors are ``presumed to be closer to the ``corners" of the data".  As far as we know, we are the first to make this explicit.  Using this observation we can prove consistency \ref{cor:consis} and convergences \ref{cor:1} of the algorithm. As an aside we give a theoretical answer to a question posed in \cite{chaudhry2006pixel}.

We extend the algorithm by also giving a method for reducing the number of hyperplanes \ref{alg:hyper}. We also briefly touch on the endmember detection problem: given a collection of points uniformly sampled from some polytope, how do we find the corners of the polytope (as opposed to the corners of the convex hull of the points).

Our paper is structured as follows. In Section 2 we explain how our method computes approximate curvature. Section 3 describes how to compute the error in our approximations. Section 4 describes our algorithm and its extensions. Section 5 uses  the results of Section 2 and 3 to show our algorithm is consistent and converges. Finally in section 6 we give some examples.

\begin{remark}
Note there is a large body of work on approximately convex bodies, see for example~\cite{barvinok2014thrifty} and \cite{ball1997elementary}.  There the idea is to, for example, find the best ellipsoid inside or containing a convex body, where a scaling factor is allowed.  The goal of this work it to select extreme vertices in a computationally practical fashion, which is somewhat different from those works. 	
\end{remark}


\section{Consistency of the curvature approximation}

Here we define the basics and show how one can compute the curvature of polytopes. This is the essence of our algorithm defined in a later section.

\begin{definition}
 A supporting hyperplane at $v$  for the convex set $P$ is a hyperplane with normal $n$ such that 
 \[
 n^\intercal y \leq n^\intercal v, \quad  \text{ for all $y\in P$}. 
 \]
\end{definition}
 Let $S^{n-1}$ be the unit sphere in $\Rn$
\begin{definition}
Let $N_v\subset S^{n-1}$ be the set of all unit normals of supporting hyperplanes at $v$ for the convex set $P$ . We view the normal vectors as being points on the sphere. The \emph{curvature} of $P$ at $v$ is the spherical volume of $N_v$. The \textit{relative curvature} of $P$ at $v$ is 
\[
K(v) = \frac{ \vol(N_v) }{ \vol(S^{n-1})}.
\]
\end{definition}

\begin{remark}
In two dimensions the curvature is a measure of the exterior angle at a point. In general, curvature is viewed as a $n$ dimensional notion of spherical angle.	 An introduction to curvature of polytopes can be found in \cite[p.~241]{pak2008lectures}. We note this definition applies just as well to any convex body~\cite{schneider2013convex}. For smooth bodies, the curvature of a subset $V$ of a convex body is $\vol ( \bigcup_{v\in V} N_v) $.  
\end{remark}

We can now show how to approximately measure the curvature of polytopes.

\begin{definition}
Let $D \subset S^{n-1}$ be finite and let $V \subset \Rn$ also be finite.  
Let 
\[ 
D_v := \{ d\in D \mid  v^\intercal d  \ge w^\intercal d  \text{ for all $w \in V$ } \} 
\] 
be the set of extremal directions in $D$ for $v$, and let 
\[
CH_{D}(V) = \{ v \in V \mid \abs{D_v} > 0 \}
\]
be the $D$-convex hull of $V$. 
Define the relative $D$-curvature of $S$ at a $D$-extremal point $v$ to be
\[
K_D(v) = \frac{\Abs{D_v}}{ \Abs{D}}
\]
\end{definition}

\begin{theorem}[Consistency of the curvature approximation]\label{thm:consistency}
Let  $D$ be a finite set of vectors uniformly sampled from $S^{n-1}$.  
Then for any $\epsilon >0 $
\[ 
\Prob\left ( \Abs{ K_D(v) - K(v) } >\epsilon \right ) \leq \frac{K(v)(1-K(v))}{ \abs{D} \epsilon ^2} 
\]
\end{theorem}

\begin{proof}

First note that since every direction vector is maximized by some vertex of a polytope we know the sets $\{ N_v \big| v\in V \}$ cover the sphere. On the other hand for $w\neq v$, $N_v \cap N_w$ has measure zero. 
This follows since for a given normal vector in $N_v \cap N_w$, the corresponding supporting hyperplane contains the line segment between $v$ and $w$. Consequently the dimension of this set must have measure zero. 

Hence, any $d_i$ has probability $K(v)$ of satisfying $d_i \in N_v$.  Moreover, for each $d_i$ belongs to $N_v$ for some $v$ (or with probability $0$ it is maximized by more than one). Therefore $\Abs{D_v}$ follows a multinomial distribution; every $d_i$ has probability $K(v)$ of belonging to $N_v$. Now, given that $K_D (v) = \frac{\Abs{D_v}}{ \Abs{D}}$, Chebyshev's inequality states that
 \[ \Prob \left ( | K_D(v) - E(K_D(v)) | >\epsilon \right ) \leq \frac{var( K_D(v))}{ \epsilon ^2} \]
and since $\Abs{D_v}$ follows a multinomial distribution $E(K_D(v))=K(v)$ and $var(K_D(v)) = K(v) (1- K(v))/\Abs{D}$. The result follows immediately. 
\end{proof}

Figure \ref{fig:curvature} shows the convergence of $K_D (v)$ for a given vertex as $|D|$ increases.

\begin{figure}[h]
\includegraphics[scale = .35]{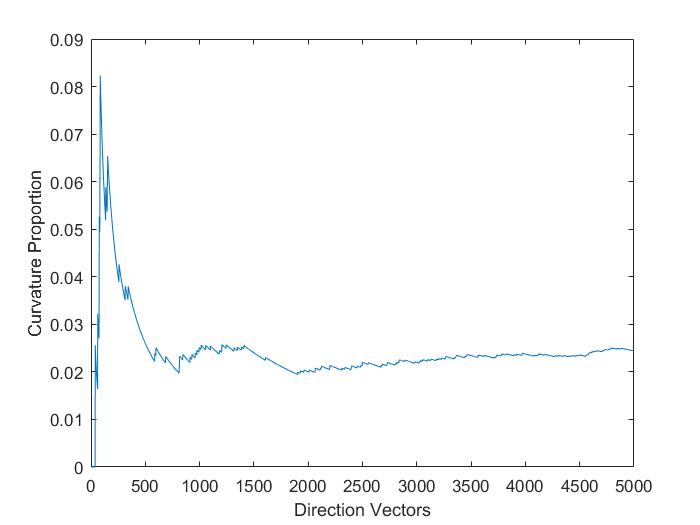}
\label{fig:curvature}
\caption{Convergence of $K_D (v)$ for a given vertex as $|D|$ increases. }
\end{figure}

\section{Sparse approximation of polytopes}

Here we show that removing low curvature extreme points from a polytope does not significantly change its shape (more precisely stated below).   The main result is related to the Aleksandrov's maximum principle \cite[p. 12]{gutierrez2012monge}

\begin{definition}
Let $d(S,S')$ denote the Hausdorff distance between two sets
\[
d(S,S') \equiv \max \{ \sup_{x\in S} d(x,S'), \sup_{y\in S'} d(y,S)   \}
\]
\end{definition}
We need a simple lemma about the Hausdorff distance between two polytopes.
\begin{lemma}\label{lem:SimpHaus} Let $S = CH(\{v_1,...v_m \})$ and let $S'$ be convex and compact then 
\[ \sup_{x\in S} d(x,S') = \sup_{1\leq i \leq m} d(v_i,S')\]
\end{lemma}
\begin{proof}
Let $W$ be a point in $S$ such that $d := d(W,S')$ is maximal. Let $A$ be a point in $S'$ such that $dist(W,A)= d$ (which exists by compactness).  Consider the hyperplane $\mathcal{H}$ at $A$ with normal $WA$. We claim this is a supporting hyperplane at $A$ with respect to the polytope $S'$. Suppose otherwise, then there exists a point $B\in S'$ that belongs to the same side of the hyperplane as $W$. But then the line $AB$ belongs to $S'$ and this line will intersect the ball of radius $d$ centered at $W$ and so there is a point in $S'$ that is closer to $W$ then $A$: contradiction. The following figure makes this clear (the circle pictured is of radius $d$).

\includegraphics{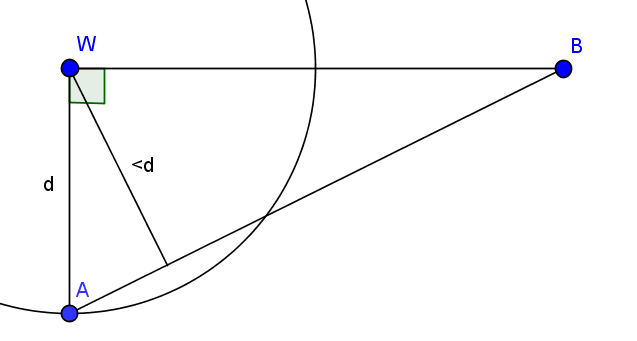}

On the other hand there must be some $w_i$ on or above (i.e. not on the same side as $A$) the hyperplane at $W$ with normal $WA$ since $w_i$ are the extreme points of $S$. If $w_i$ is on the hyperplane then $d(w_i,S')\geq d$ and we are done. Otherwise, if it is above, then since $d(w_i,S')\leq d$ there exists a point $A_0$ in $S'$ such that $d(W,A_0) \leq d$ and so $A_0$ is on the wrong side of $\mathcal{H}$ giving the above contradiction.
\end{proof}
 
	Recall for an angle $\theta \in [0,\pi ]$ and $v\in S^{n-1}$ a spherical cap with angle theta about $v$ is the set 
	\[ \{w\in S^{n-1} \big| w^\intercal v \geq \cos (\theta ) \} \]
	
	Let $S_{cap}(\theta)$ denotes the volume of the $n$ dimensional spherical cap with angle $\theta$ (about any $v$ since this will not change the volume). 
	
	\begin{lemma}. \label{lem:ball} Let $\theta \in [0,\pi ]$, then in $S^{n-1}$
	\[ \frac {1}{2} \left( \sin( \frac {\theta}{2} ) 
\right)^{n-1} \leq S_{cap} (\theta) 	\]
	\end{lemma}
\begin{proof}
 A straightforward rewriting of  \cite[lemma 2.3]{ball1997elementary}
\end{proof}

Let $B_r$ denote the ball of radius $r$ in $\Rn$. 

\begin{theorem}[Aleksandrov's maximum principle]\label{thm:Alex} Let $V = \{v_1,...v_m \}$ and $W = \{w_1,...,w_k,\}$ be subsets of  $B_r$. Let $S = CH(V \cup W)$ and $S' = CH(V)$, suppose neither are degenerate.  Let $\omega$ be the sum of the curvatures of all the points $w \in W$. Then 
\begin{equation}\label{spherical.cap}
S_{cap}\left(\arcsin\left(\frac{d(S,S')}{2r}\right)\right)\leq\omega 	
\end{equation}
and 
\begin{equation}\label{distance.estimate}
	d(S,S') \le \sqrt{2} \pi r(2\omega)^\frac{1}{n-1}
\end{equation}
\end{theorem}

\begin{remark}
The result above gives an upper bound on the distance of order $\left( \omega \vol (S^{n-1})\right)^{\frac{1}{n-1}} $ where $\omega$ is the total relative curvature removed from the set. In Figure~\ref{fig:2} we plot this function for $n=2,3,4,5$ and $r = 1$. Note since $2$ is a trivial worst case this shows how the estimate is only useful for small $\omega$
\end{remark}
\begin{figure}
\includegraphics[scale= .35]{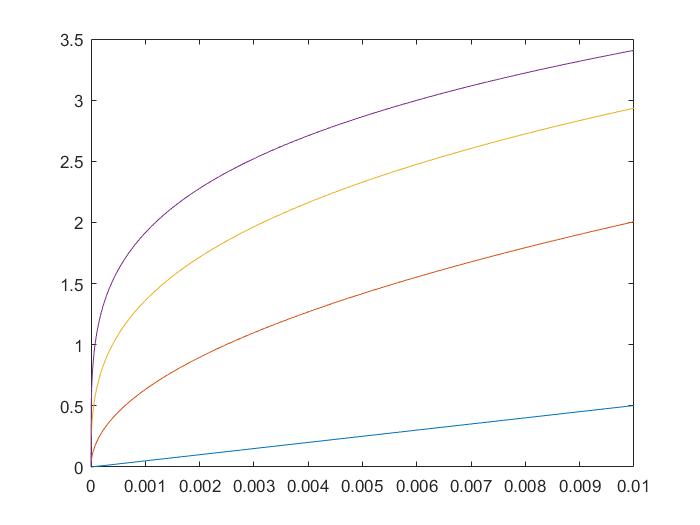}
\caption{  Plot of the worst case distance between the sets  as a function of the relative curvature removed, $\omega$, for $n=2,3,4,5$ and $r = 1$ (the error is increasing in $n$). 
$\sqrt{2} \pi r\left (2\omega\vol (S^{n-1})\right)^\frac{1}{n-1}$ }
\label{fig:2}
\end{figure}

\begin{proof}[Proof of theorem \ref{thm:Alex}]

Let $d = d(S,S')$. First we show how to obtain \eqref{distance.estimate} from \eqref{spherical.cap}.   From lemma \ref{lem:ball}
\[
\frac {1}{2} \left( \sin \left( \frac{\arcsin(\frac{d}{2r})}{2} \right) \right)^{n-1} \leq S_{cap}\left(\arcsin \left( \frac{d}{2r} \right) \right) \leq \omega 
\] 
hence: 
\[ \sin \left( \frac{\arcsin(\frac{d}{2r})}{2} \right) \leq 2 (2\omega)^\frac{1}{n-1} \]
Finally we use the inequalities
 \[\sin \left( \frac{\arcsin(\frac{d}{2r})}{2} \right) \geq \frac {4}{\sqrt{2}\pi } \left(\frac{\arcsin(\frac{d}{2r})}{2}\right)\] 
 and 
\[ \frac {4}{\sqrt{2}\pi } \arcsin \left( \frac{d}{2r} \right) \geq \frac {4}{\sqrt{2}\pi } \left(\frac{d}{2r}\right)\] 
to complete the proof of this step.

Next we prove \eqref{spherical.cap}.  First we show it suffices to consider the case where $k=1$.  Suppose we proved the theorem for $k=1$. Let $w_i$ be the point in $S$ furthest away from $conv(\{v_1 ,...v_m\})$ which exists by Lemma \ref{lem:SimpHaus}. But by assumption, $S_{cap}\left(\arcsin\left(\frac{d(w_i,S')}{2r}\right) \right) \leq \omega'$  where $\omega'$ is the curvature of $w_i$ with respect to $S'':=conv(w_i,v_1,...v_n)$. Thus it would be enough to show $\omega' \leq \omega$. But note every supporting hyperplane for some $w_j$ with respect to $S$ will either be a supporting hyperplane for some $v_\ell$ or for $w_i$ with respect to $S''$, on the other hand any supporting hyperplane for some $v_j$ with respect to $S$ will remain a supporting hyperplane of $v_j$ with respect to $S''$. From this we conclude that $\omega' \leq \omega$ as required. 

It remains to show the case $k=1$. Rename $w_1$ to $W$ for clarity, so $S=conv(W,v_1,...v_n)$, note $d = d(W,S')$.  Let $A$ be the unique point in $S'$ such that $dist(A,W)=d$ (which exists by compactness of convex hulls). Let  $\mathcal{H}$ be the hyperplane at $A$ with normal $WA$. Note that exactly as in the proof of  \ref{lem:SimpHaus} $\mathcal{H}$ is a supporting hyperplane for $S'$

Consider the cone $C_0$ consisting of vertex $W$ and base $S\cap \mathcal{H}$. In other words the portion of $S$ on the same side of the hyperplane as $W$. Let $C_1$ be the cone consisting of vertex $W$ and base the ball of radius  $\sqrt{ (2r)^2 - d^2}$ centered at $A$ in  $\mathcal{H}$. We claim $C_0\subset C_1$, it suffices to show $S\cap \mathcal{H}$ is contained in the ball around $A$. Note by assumption all points of $S$ are in a ball of radius $r$, in particular all points of  $S\cap \mathcal{H}$ are within $2r$ of $W$. Moreover the slant of $C_1$ is length $2r$. So if any point in $S\cap \mathcal{H}$ were outside the ball around $A$ it would be a distance $>2r$ from $W$: contradiction.

Finally, since $C_0\subset C_1$ it follows that the curvature of $W$ with respect to $C_0$ is less then or equal to the curvature of $W$ with respect to $C_1$ but this is precisely:
\[ 
S_{cap}\left(\arcsin ( \frac{d}{2r})\right) \leq\omega 
\qedhere
\]

\end{proof}


\section{The algorithm}

Algorithm \ref{alg:Basic} is our basic approximate convex hull algorithm. In short, we generate many direction vectors and use these to compute curvature as described above. We then keep only the high curvature points.
\begin{algorithm}
\label{alg:Basic}
\SetKwInput{Input}{Input}
\SetKwInput{Output}{Output}

\caption {Basic Algorithm}
\Input  {A finite collection of points $V\subset \mathbb{R} ^n $ and a threshold $0\leq \alpha\leq 1$. A set  $D\subset S^{n-1}$ of direction vectors, generated uniformly from the unit sphere (see below for some alternatives)}
\Output {$V' \subset V$ an approximate convex hull (called the inner hull). modified slightly this can output a collection of constraints that determine an approximate convex hull (called the outer hull), see below for details. }
\Begin
{
	\ForAll {$d \in D$}
	{
		\ForAll {$v \in V$}
			{
			$a_{v,d} \leftarrow v^\intercal d$
			}
			Find $w$ such that $a_{w,d}$ is max
			
			$Count_w \leftarrow Count_w + 1$
	}

\ForAll {$v \in V$} 
{$K_D (v) \leftarrow Count_v / |D|$

\If {$K_D (v) \leq \alpha$} {Delete $v$ from $V$ 

(Optional) Alternatively for everything below the threshold remove the $v$ `proportionally'. For example we may remove with probability $1 - \frac{K_D(v)}{\alpha}$.}

}
The output is the remaining elements of $V$ denoted $V'$
}
\end{algorithm}

\begin{definition}
Suppose we run algorithm \ref{alg:Basic} on sets $V$ and $D$ as above. Let $V'$ be all the vectors that were kept. We call $CH(V')$ the \emph{inner hull}. Now for each $d\in D$ let $v_d\in V$ be a vector that maximizes the dot product. Consider the collection of linear constraints 
\[d^\intercal x \leq d^\intercal v_d\]
This determines a convex body we call the \emph{outer hull}. 
\end{definition}
Note the inner hull is contained in the actual convex hull of $V$, which is contained in the outer hull. Our algorithm also gives the constraints for the outer hull. In step 2 one simply needs to keep track of the value $d ^\intercal v_d$ (i.e $a_{v,d}$) for each $d\in D$ when computing the maximums. There is a possibility the constraints will not define a finite polytope if $|D|$ is too small but this is exceedingly unlikely for large 
$|D|$. An even larger concern is that the outer hull contains a large number of constraints, we will attempt to remedy this later. 

We claim that both the inner and outer hull approach the actual convex hull if one increases the number of direction vectors. Thus we get an approximate convex hull in the vertex format and another approximate convex hull in the constraint format. The following notion of error will help make this precise. 

\begin{definition}
Under the same assumptions as the previous definition let $A$ be the inner hull, $B$ the actual hull and $C$ the outer hull. We define the \emph{inner error} as 
\[ \sup_{x\in B} d(x,A)\] 
and the \emph{outer error} as 
\[ \sup_{x\in C} d(x,B) \]
In both cases these are simply the Hausdorff distance defined above.
\end{definition}
\begin{remark}[Generating direction vectors]

In \cite{theiler2000using} we find an elegant method for producing direction vectors (called skewers). Not only does this speed up the production of direction vectors but more importantly it speeds up the computation of the dot products in the algorithm.  We are confident choosing vectors uniformly from the sphere works (Since it gives a precise measure of the curvature as we saw above). However alternative methods seem to work just as well experimentally (intuitively they are `uniform enough'). In \cite{chang2008pyramid} they have tested various methods similar to \cite{theiler2000using} and have singled out what they found to be the best approach. 
\end{remark}

\subsection{Sparse Approximation}
For many applications the above algorithm finds too many extreme vectors in the inner hull, and for almost any application the outer hull has far too many hyperplanes (there will be one hyperplane for each direction vectors). In this section we discuss how to deal with this problem, in particular how to reduce the following two ratios:

\begin{definition}
Let $V \subset \mathbb{R}^n$ be a finite collection of points. Suppose we have an algorithm $A$ that outputs vertices of a convex hull $A_{CH}$ that approximates $CH(V)$. Then the \emph{vertex compression ratio of algorithm $A$} is
\[ \frac{\text{Number of vertices in } A_{CH}}
{\text{Number of vertices in }  CH(V)}\]
On the other hand suppose $A$ outputs hyperplanes of a convex hull $A_{CH}$ that approximates $CH(V)$. Then the \emph{hyperplane compression ratio of algorithm $A$} is 
\[ \frac{\text{Number of hyperplanes in } A_{CH}}
{\text{Number of hyerplanes in } CH(V)}\]

\end{definition}

To reduces the vertex compression ratio of algorithm \ref{alg:Basic}  we run algorithm \ref{alg:cluster}. We simply find vectors that are clustered together and keep only one of them. 

\begin{algorithm}
\label{alg:cluster}
\SetKwInput{Input}{Input}
\SetKwInput{Output}{Output}

\caption {Vertex Compression Algorithm}
\Input  {The inputs and outputs of algorithm \ref{alg:Basic} and new threshold $\beta\geq 0$}
\Output {$V'' \subset V'$ an approximate convex hull}
\Begin
{
  Order $V':=\{ v_1,v_2...v_m\}$ in increasing order of relative $D$-curvature $K_D(v)$
  
 $i \leftarrow 1$
 
	\While {$i \leq m$}
	{
		\ForAll {$w \in V'$}
			{
			\If {  $d(v_i,w) < \beta$ } {remove $w$ from $V'$.}
			}
	Set $i$ to the next largest interger for which $v_i$ is not already removed from $V'$		
	}
The output is the remaining vectors of $V'$, call this $V''$
}
\end{algorithm}

\begin{remark}
Let $V,V',V'',\beta$ be as in algorithm \ref{alg:cluster} then 
\[d(CH(V'),CH(V'')) < \beta\]
\end{remark}
\begin{proof}
This is an immediate application of Lemma  \ref{lem:SimpHaus}
\end{proof}
Next we show how to reduce the hyperplane compression ration of our original algorithm. Let $V,D$ be as above, we start by running our original algorithm followed by the vertex compression algorithm. Suppose when running our original algorithm that for each $v\in V'$ we keep track of 
\[D_v := \{ d\in D \mid  v^\intercal d  \ge w^\intercal d  \text{ for all $w \in V$ } \} \] 
moreover suppose when running the vertex compression algorithm that for each $v\in V''$ we keep track of 
\[E_v := \{ w \in V' \big| w \text{ was removed during the step involving } v \}\cup\{v\}\]
(i.e. $E_v$ is all elements in $V'$ that were clustered around $v$). For $v\in V''$ let 
\[F_v := \cup_{w\in E_v} D_v \]

See algorithm \ref{alg:hyper} below.
\begin{algorithm}[ht] 
\label{alg:hyper}
\SetKwInput{Input}{Input}
\SetKwInput{Output}{Output}

\caption {Hyperplane Compression Algorithm}
\Input  {The inputs and outputs of  previous algorithms including values mentioned above \ref{alg:Basic}}
\Output {Constraints that determine an approximate convex hull to $V$ that contains $V$. There is a possibility the constraints will not define a convex hull if $|D|$ is too small.}
\Begin
{
 $F \leftarrow \emptyset$
  
		\ForAll {$v \in V''$}
			{
			Project $F_v$ into $\mathbb{R}^n$
			
			Run algorithm \ref{alg:Basic} and \ref{alg:cluster}
			
			Look at the output of \ref{alg:cluster}, and the corresponding unprojected vectors from $F_v$ into $F$
			
			}
			
			(Optional) Alternatively we could choose a threshold $\gamma$. For each $v\in V''$ and $d\in D$ compute $v^\intercal d$. If for any $d$ we find that there exists three distinct points $\{ v_1 , v_2, ..., v_n\}\subset V''$ such that \[ |v_i^\intercal d - v_j^\intercal d | < \gamma \] for all $1\leq  i,j \leq 3$ then we put $d$ into $F$. Unless stated otherwise we will assume this method is not used
			
			Find clusters of points in $F$. (collections of points that are close together)
			
			Replace each cluster with the weighted sum of elements in that cluster to get $F'$
			
			Run algorithm \ref{alg:Basic} with $\alpha = 0$ and $D = F'$ finding the outer hull. This is the output. 
			}
\end{algorithm}

If desired, algorithm \ref{alg:hyper} works without running the vertex compression algorithm (i.e. if $\beta = 0$). If one doesn't run the vertex compression algorithm \emph{and} finds the true convex hull of the projected $F_v$ then the result is equivalent to the outer hull found originally (this merely removes redundant hyperplanes)

We end this section by noting that hyperplane compression algorithm is a potential solution to the problem of finding endmembers (the vertices of the outputed hull are potential endmembers). This holds even without a pure pixel in the data set. In particular if the data is uniformly generated from some polytope and we wish to recover the polytope this method can be used.

\section{Application of the curvature estimate to the algorithm}

We can use the consistency result, Theorem~\ref{thm:consistency}, to show convergence of algorithm~\ref{alg:Basic}.
\begin{corollary}\label{cor:1}
Let $V \subset \mathbb{R}^n$ and let $D$ be uniformly sampled from $S^{n-1}$.  Assume that the $C(v) \geq \omega > 0$ for all extremal points in $V$.  
Then
\[
CH_D{V} = CH(V), \quad  \text{ with probability $P \ge 1-p$}
\]
provided 
\[
\abs{D} \geq\frac {\log (\omega p)}{\log(1-\omega)}.
\]
\end{corollary}

\begin{proof} 
We have seen from the proof above that the probability of not finding a point with relative curvature $\omega$ is $(1-\omega)^{\Abs{D}}$.Since there are at most $\frac{1}{\omega}$ such points we have by subadditivity of probabilities that the probability of missing one of them is $\leq \frac{1}{\omega} (1-\omega)^{\Abs{D}}$ (for an exact answer use inclusion-exclusion) Suppose $p$ is less then or equal to this: $p\leq \frac{1}{\omega} (1-\omega)^{\Abs{D}}$. This is trivially equivalent to $\Abs{D}\geq\frac {\log (\omega p)}{\log(1-\omega)}$ as required
\end{proof}

\begin{remark}
	The function $\frac {\log (\omega p)}{\log(1-\omega)}$ in the estimate above is $O(1/\omega)$ as case be seen in   Figure~\ref{fig:cor}.
\end{remark}

\begin{figure}[h]
\includegraphics[scale= .35]{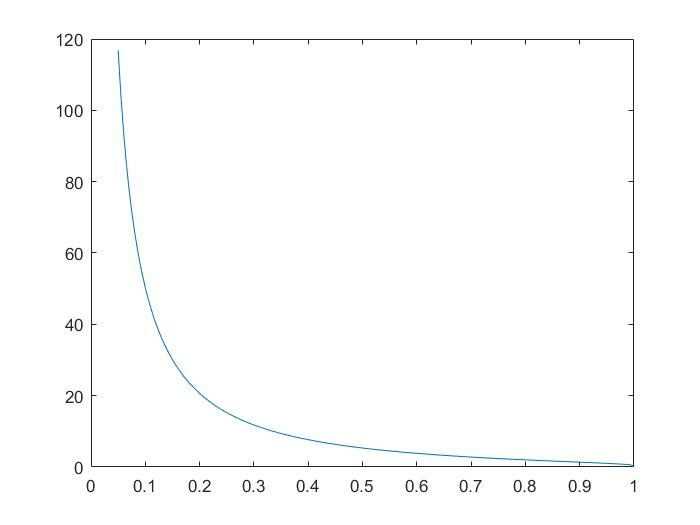}
\caption{Illustration of the Corollary~\ref{cor:1}. 
The  number of direction vectors required to find all points with relative curvature $x$ with probability 95\% i.e. $y=\frac {\log ( .05x)}{\log(1-x)}$ }
\label{fig:cor}
\end{figure}

Another consistency result is as follows.

\begin{corollary}\label{cor:consis}
Let $V \subset \mathbb{R}^n$. Let $D$ be sampled uniformly from $S^{n-1}$. 
Let  $O_k$ and $I_k$ be the outer and inner error respectively. Then $\{O_k \}_{k\in\mathbb{N}}$ and $\{I_k \}_{k\in\mathbb{N}}$ are non increasing sequences and with probability $1$ they both converge to zero.
\end{corollary}
\begin{proof}
It is clear that $\{O_k \}_{k\in\mathbb{N}}$ and $\{I_k \}_{k\in\mathbb{N}}$ are non increasing sequences since each new direction vector adds one more constraint (possibly lowering $O_k$ and may or may not find a new extreme vector (possibly lowering $I_k$).

It is also clear from above that $I_k$ approaches $0$. Indeed each extreme vector has non zero curvature and hence a non zero probability of being found so for large enough $k$ we expect $I_k = 0$ with increasing probability.

Finally we can use theorem \ref{thm:Alex} to show with probability $1$ that $O_k$ converges to zero. As before we know that with probability $1$ all extreme points will be found. The extreme points of the outer hulls are made up of $V$ and some other vectors. For each extreme point $v\in V$, we claim the curvature of $v$ with respect to the actual hull $CH(V)$ approaches the curvature of $v$ with respect to the outer hull as $k\rightarrow \infty$ with probability $1$. To show this let $E$ be the set of all sequences $(d_1,d_2,...) \in (S^{n-1})^\infty$ such that the claim fails. We wish to show the measure of $E$ is zero. Let $N_v \subset S^{n-1}$ be the collection of direction vectors for which $v$ is extremal and let $\{ d_1 ',d_2 ' ... \}$ be the subsequence of $d_i$ that belong to $N_v$. Now if each point in $N_v$ is a limit point of $\cup_n CH(\{ d_1 ' ,...d_n '\} )$ the claim clearly follows. Therefore for each sequence in $E$ there exists a rational point $q\in N_v$ and a rational $r>0$ such that $B_r(a)\cap \cup_n \CH(\{ d_1 ',...d_n '\} ) = \emptyset$. Let $E_a,r$ be all sequences for which this holds, now since $N_v \cap B_r(a)$ has positive measure then by definition of product measure the measure of $E_a,r$ is zero. Since $E$ belongs to the union of all $E_a,r$ the claim follows from countable subadditivity of measure.

 Now since the sum of the curvatures of each extreme point of any polytope add up to $\vol(S^{n-1})$ the claim implies that all extreme points of the outer hulls that are not in $V$ must have a vanishing proportion of the total curvature as $k\rightarrow \infty$. Theorem  \ref{thm:Alex} completes the proof.
\end{proof}
It is easy to compute some controls on the convergence of the inner error. Equivalenty this gives a worst case calculation for choosing the appropriate number of  direction vectors to achieve a desired error. This is a \emph{theoretical} answer to a problem raised in \cite{chaudhry2006pixel}, how to choose the number of direction vectors for PPI algorithms. It's worth mentioning the practical solutions from \cite{chaudhry2006pixel} for choosing the number of direction vectors. Essentially they suggest computing the maximum for only a small block of direction vectors and then repeating this process until no new extreme vectors are found. This way there is no need for human input about the choice of how many direction vectors to use. 

\begin{remark}{\emph{Number of direction vectors needed for a given inner error (simple messy bound)}}
We can use the above to find the number of direction vectors needed to achieve a 
particular inner error $\epsilon >0$. Let $V$ be a set of points in ball of radius $r$ and let $X$ be the true extreme points of $\CH(V)$.   Suppose we run our algorithm with some set of direction vectors uniformly chosen and we keep all extreme points found.  The error depends on the total amount of curvature of all the extreme points we've missed. Using \ref{cor:1} we can ensure with probability $\geq 1-p$ that we have all points of relative curvature $\geq \omega$ for any $0\leq \omega \leq 1$. Then the total missing curvature would be $\leq |X|\omega \vol(S^{n-1})$. By \ref{thm:Alex} this will give an error 
\[\leq \sqrt{2} \pi r\left (2|X|\omega \vol(S^{n-1})\right )^\frac{1}{n-1}\]. 

We can set this error $\leq \epsilon$ and solve for $\omega$ to get 
\[ \omega \leq \frac{(\frac{\epsilon}{\sqrt{2} \pi r})^{n-1}} { 2|X| \left (\vol(S^{n-1})\right )}\]
Denote the right hand side by $C$. From the result above to achieve this $\omega$ we would require 
\[\frac {\log (Cp) }{ \log (1-C)}\]
direction vectors. To sum up, if the number of directions is more than $\log (Cp) / \log(1-C) $ then with probability $1-p$ the error is bounded by $\epsilon$. 

The growth of this is comparable to $\frac{\vol(S^{n-1})}{(\frac{\epsilon}{\sqrt{2} \pi r})^{n-1}}$. This unfortunately gives results that are very large. For example keeping $p=.05$,$\epsilon = .1$, $r=1$,  $|X|=10000$,  and varying the dimension $n = 3,4,5,6,7$ we get roughly $10^{10},10^{11},10^{13},10^{15},10^{17}$ respectively.
\end{remark}

\begin{remark}{\emph{Improvement on previous remark}}
The previous remark doesn't relate to the inner error we expect in practice. For one thing, of course we don't claim the constants found above are in any way optimal but more importantly if we remove many low curvature points we don't expect this to be equivalent to the worst case where all the curvature is concentrated in one point. This is especially true since we are taking our directions uniformly from the sphere. In many cases if we have found all points with relative curvature $\geq \omega$ then removing all other points results in an error closer to $\leq  \sqrt{2} \pi r (2 \omega \vol(S^{n-1}) ) ^\frac{1}{n-1}$ (the curvatures don't add up). This means we can set $|X|=1$ in the calculation above. So if we consider as before  $p=.05$,$\epsilon = .1$, $r=1$, and vary the dimension $n = 3,4,5,6,7$ we get roughly $10^5,10^7,10^9,10^{11},10^{13}$ respectively. 

\end{remark}

\section{Computational Examples}
\subsection{Synthetic data}
In this section we generate synthetic data to test the algorithm. 

For points generated uniformly from a simplex most curvature is in the corners.  So already a random uniform direction vector will likely find a corner. Our procedure is effective in this case. 
\begin{figure}[h] 
\includegraphics[scale= .35]{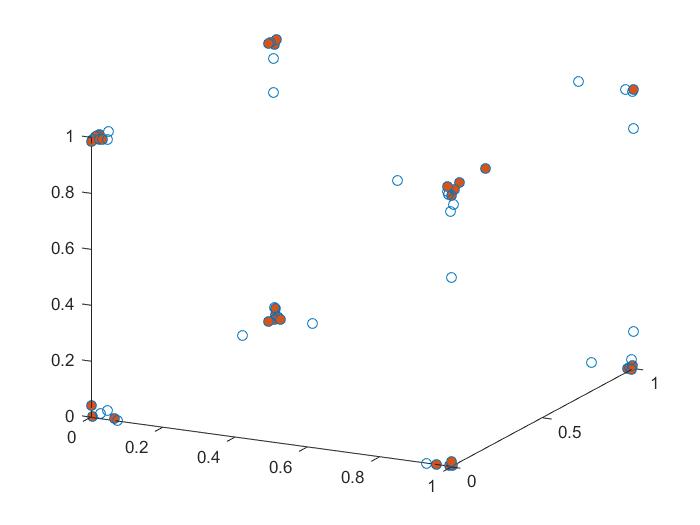}
 \caption{Recovering the vertices of the cube in three dimensions.}
\label{fig:cube}
\end{figure}

Consider Figure \ref{fig:cube}. This is 1000 direction vectors and a million points. The algorithm found 61 extreme vectors (unfilled dots) and kept 34(filled dots). The actual convex hull has around 300 extreme points. Clearly clustering (algorithm \ref{alg:cluster} would be very effective in this case.

On the other hand for a million points generated uniformly from the sphere we do very poorly (Figure \ref{fig:sphere}). Here all the direction vectors found different extreme points. Only 102 were kept. (This can be slightly rectified by modifying step 4 at the expense of worse performance on examples with sharp corners).

\begin{figure}[h]
\includegraphics[scale= .25]{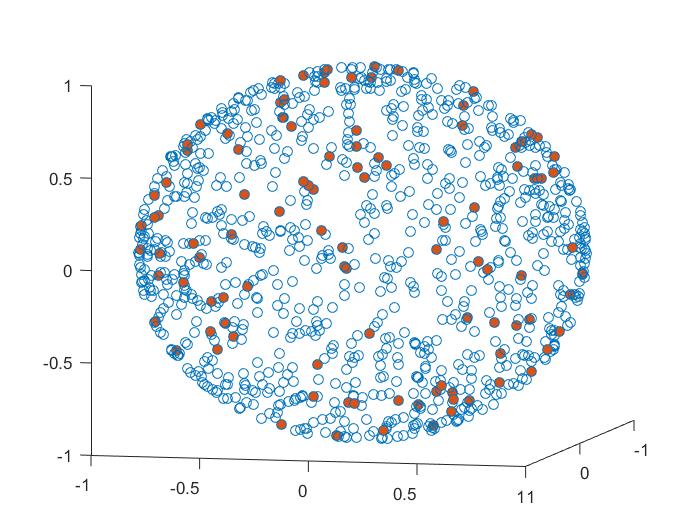}\includegraphics[scale= .25]{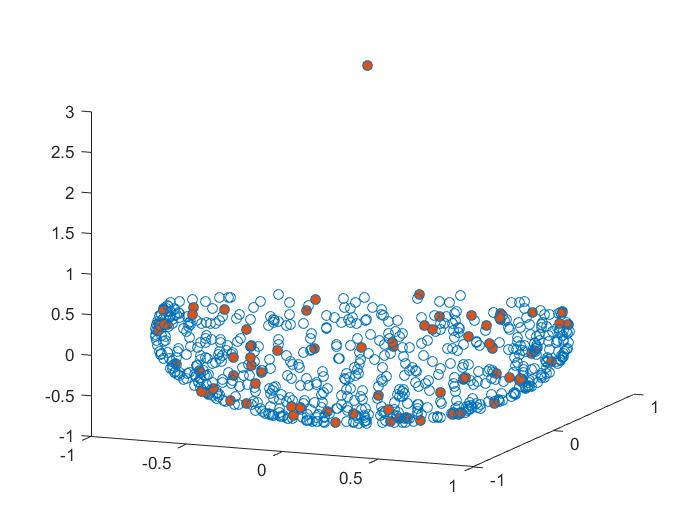}
\caption{Sphere and Ice Cream Cone (isolated extreme point).}
\label{fig:sphere}\label{fig:cream}
\end{figure}

Figure \ref{fig:cream} is an example of a mix of low curvature points with one very high curvature point. In practice much data is such a mix, so we expect high performance in some region and low performance in others.

Figure~\ref{fig:cubedat} studies the accuracy of the algorithm as a function of the number of direction vectors. We generated 10 thousand 3 dimensional points and applied the algorithm with an increasing number of direction vectors, keeping every extreme point we found. The figure shows how the outer error (yellow) and the inner error (blue) decrease as the number of found extreme vectors increases. The points in the first two images were generated randomly from the cube and from the sphere. The points in the third  randomly generated points from a simplex, and the we applied a fixed linear transformation. (Note for the sphere the outer error is only computed to the first digit due to the large number of computations involved).  Figure~\ref{fig:cubedat}(d) shows a similar computation of the inner error using 1 million points in 3 dimensions. This time the error is compared to the number of extreme vectors found.

\begin{figure}
\includegraphics[scale = .35]{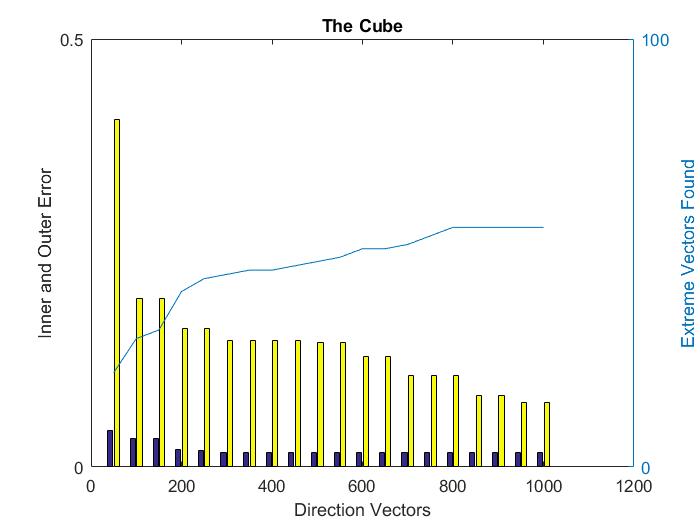}\includegraphics[scale = .35]{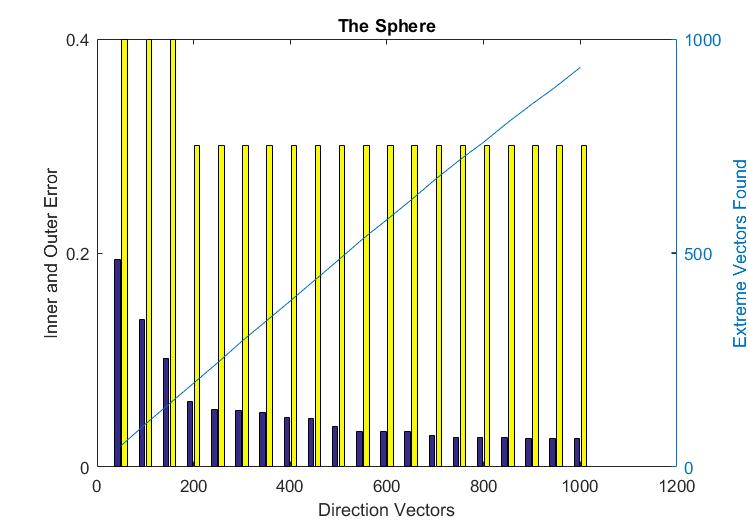} 
\includegraphics[scale = .35]{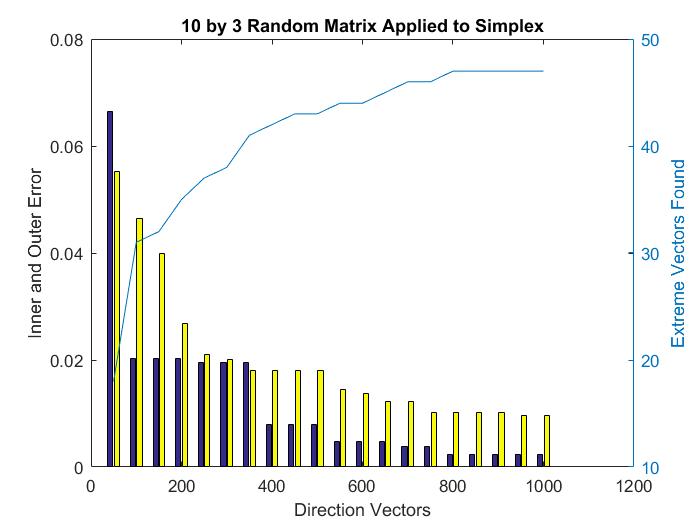}\includegraphics[scale = .35]{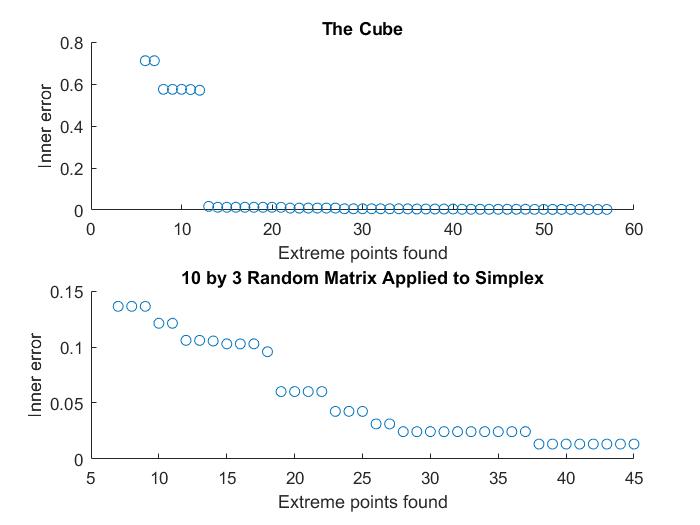}
\caption{Comparison of inner and outer error, along with number of extreme points found. Points sampled uniformly from 3 dimensional (a) cube (b) sphere (c) linearly transformed simplex using 10,000 points.  (d) Inner error using 1 million points. }
\label{fig:cubedat}
\end{figure}

%
%

\subsection{Algorithm Applied to Helicopter Flight Test data set.}
In this section, we demonstrate the results of the vertex compression algorithm on a real dataset consisting of 23 million data points in dimension 5. We used 70000 direction vectors  In Figure~\ref{fig:vertcomp_allfound}(a) we compare the 1491 found points (red filled) with the 2495 points on the true convex hull. There is an error of around 300, for context the width of the shape is around 66000. Moreover the average distance from the true extreme points to the mean of the extreme points is 32000. The figure has been projected into 2 dimensions.    After  running the vertex compression algorithm we were left with 29 points \ref{fig:vertcomp_allfound}(b). The error here is 6800, for context recall the information above.  

\begin{figure} [ht]
\includegraphics[scale= .3]{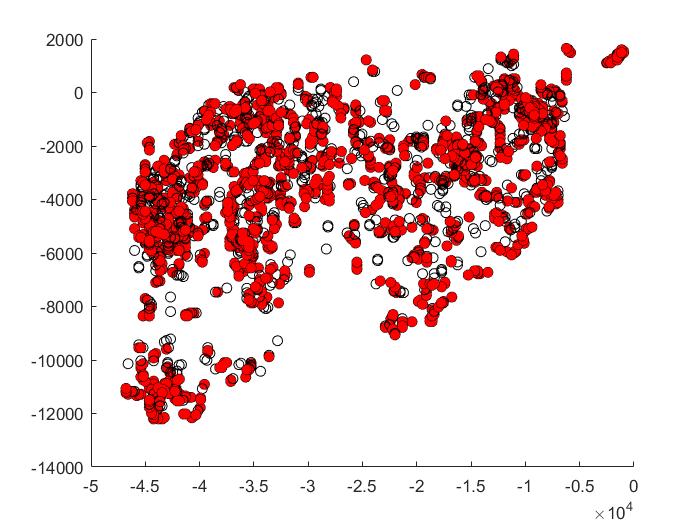}\includegraphics[scale= .3]{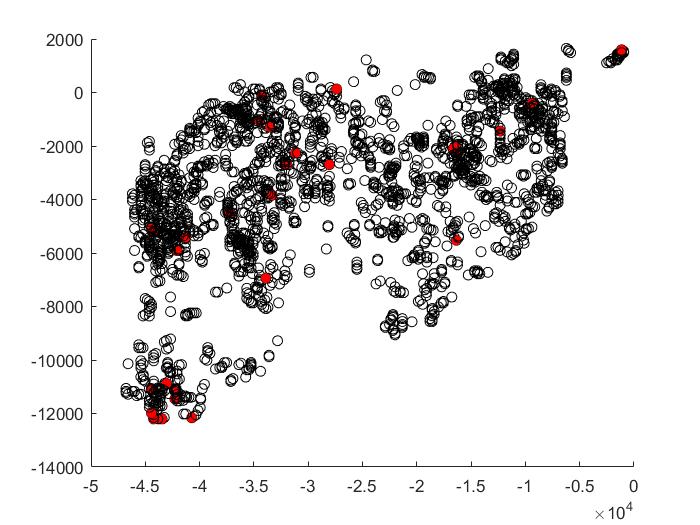}
\caption{Left: the actual hull versus the points found.  Right: the final 29 extreme vertices.}
\label{fig:vertcomp_allfound}
\end{figure}

%

\bibliographystyle{alpha}

\bibliography{ACH}

\end{document}